\documentclass[11pt]{article}
\newcommand{\PAPER}[1]{#1}
\newcommand{\LIPICS}[1]{}

\PAPER{
\usepackage{fullpage,amsthm,amssymb,hyperref}

\title{Dynamic Geometric Data Structures via Shallow Cuttings}
\author{Timothy M. Chan\thanks{Department of Computer Science, University of Illinois at Urbana-Champaign (tmc@illinois.edu).  Work supported in part by NSF Grant CCF-1814026.}}

\newtheorem{lemma}{Lemma}[section]
\newtheorem{theorem}[lemma]{Theorem}
\newtheorem{corollary}[lemma]{Corollary}

}

\newcommand{\R}{\mathbb{R}}
\newcommand{\eps}{\varepsilon}

\newcommand{\LE}{\mbox{\rm LE}}
\newcommand{\UH}{\mbox{\rm UH}}
\newcommand{\D}{\Delta}

\begin{document}
\maketitle

\begin{abstract}
We present new results on a number of fundamental problems about dynamic geometric data structures: 

\begin{enumerate}
\item We describe the first fully dynamic data structures with sublinear amortized update time for maintaining (i)~the number of vertices or the volume of the convex hull of a 3D point set, (ii)~the largest empty circle for a 2D point set, (iii)~the Hausdorff distance between two 2D point sets, (iv)~the discrete 1-center of a 2D point set, (v)~the number of maximal (i.e., skyline) points in a 3D point set.  The update times are near $n^{11/12}$ for (i) and (ii), $n^{7/8}$ for (iii) and (iv), and $n^{2/3}$ for (v).  Previously, sublinear bounds were known only for restricted ``semi-online'' settings [Chan, SODA 2002].
\LIPICS{\medskip}
\item We slightly improve previous fully dynamic data structures for answering extreme point queries for the convex hull of a 3D point set and nearest neighbor search for a 2D point set.  The query time is $O(\log^2n)$, and the amortized update time is $O(\log^4n)$ instead of $O(\log^5n)$ [Chan, SODA 2006; Kaplan et al., SODA 2017].
\LIPICS{\medskip}
\item We also improve previous fully dynamic data structures for maintaining the bichromatic closest pair between two 2D point sets and the diameter of a 2D point set.  The amortized update time is $O(\log^4n)$ instead of $O(\log^7n)$ [Eppstein 1995; Chan, SODA 2006; Kaplan et al., SODA 2017].
\end{enumerate}
\end{abstract}

\section{Introduction}

\LIPICS{\subparagraph*{Background.}}
\PAPER{\paragraph{Background.}}
\emph{Dynamic\/} data structures that can support insertions and deletions of data have been a fundamental topic in computational geometry since the beginning of the field.  For example, in 1981 an early landmark paper by Overmars and van Leeuwen~\cite{OveLee} presented a fully dynamic data structure for \emph{2D convex hulls\/} with $O(\log n)$ query time and $O(\log^2n)$ update time; the $\log^2n$ bound was later improved in a series of work \cite{FOCS99,BroJac,SoCG11} for various basic types of hull queries,
e.g., finding extreme points along given directions.

One of the key results in the area is the author's fully dynamic
data structure for \emph{3D convex hulls}~\cite{SODA06}, which was the
first to achieve polylogarithmic query and update time for basic types of hull queries.
The original solution required $O(\log^2n)$ query time for extreme point
queries, and $O(\log^6n)$ amortized update time.  (A previous solution
by Agarwal and Matou\v sek~\cite{AgaMat95} had $O(n^\eps)$ query or update time for an arbitrarily small constant $\eps>0$.)
Recently Kaplan et al. \cite{KapSODA17} noted a small modification of the data structure, improving the update time to $O(\log^5n)$.
The result has numerous applications, including dynamic \emph{2D nearest or farthest neighbor search\/} (by the standard lifting map).
Another application is dynamic \emph{2D bichromatic closest pair\/} (i.e., computing
$\min_{p\in P}\min_{q\in Q}\|p-q\|$ for two planar point sets $P$ and $Q$) or
dynamic \emph{2D diameter} (i.e., computing $\max_{p\in P}\max_{q\in P}\|p-q\|$
for a planar point set $P$):
Eppstein~\cite{Epp} gave a clever, general technique reducing
dynamic closest/farthest pair problems to dynamic nearest/farthest
neighbor search, which increased the update time by a $\log^2n$ factor;
when combined with the above, this yielded
an $O(\log^7n)$ update time bound.

For many other problems, polylogarithmic update time appears more difficult, and getting
sublinear update time is already challenging.
For example, in SoCG 2001, the author~\cite{SoCG01} obtained
a dynamic data structure for the \emph{width\/} of a 2D point set with
$O^*(\sqrt{n})$ amortized update time.\footnote{
Throughout the paper, we use the $O^*$ notation to hide small
extra factors that are polylogarithmic, or in some cases, $o(n^\eps)$
for an arbitrarily small constant $\eps>0$.
}
(Part of the difficulty is
that the width problem is neither ``decomposable'' nor ``LP-type''.)
Sublinear update time is known for a few other assorted geometric problems, such as \emph{dynamic connectivity\/} for the intersection graph of geometric objects~\cite{CPR}.

In SODA 2002, the author~\cite{SODA02} explored still more challenging dynamic geometric problems, including maintaining
\LIPICS{\medskip}
\begin{enumerate}
\item[(i)] the number of vertices and facets of a 3D convex hull,
or its volume, 
\item[(ii)] the \emph{largest empty circle\/} for a 2D point set (with
center restricted to be inside a fixed triangle), 
\item[(iii)] the \emph{Hausdorff
distance\/} for 2D point sets $P$ and $Q$ (i.e., computing $\max_{q\in Q}\min_{p\in P}\|p-q\|$
for two planar point set), and 
\item[(iv)] the
\emph{discrete 1-center\/} of a 2D point set $P$ (i.e., computing $\min_{q\in P}\max_{p\in P}\|p-q\|$).
\end{enumerate}
\LIPICS{\medskip}
The paper~\cite{SODA02} obtained sublinear results only for the \emph{insertion-only\/} case and the \emph{off-line\/} case
(where we are given the entire update sequence in advance), or a generalization of both---the \emph{semi-online\/} case (as defined by Dobkin and Suri~\cite{DobSur},
where we are given the deletion time of an element when it is inserted).
The update time bounds were $O^*(n^{7/8})$ for (i) and (ii), and
$O^*(n^{5/6})$ for (iii) and (iv).  

None of these four problems are ``decomposable''.  In particular, problem (i) is nontrivial since known
methods such as~\cite{SODA06} for 3D convex hull queries  do not
maintain the global hull explicitly, unlike Overmars and van Leeuwen's original data structure for 2D convex hulls.  Problem (ii) also seems to require explicit
maintenance of a 3D convex hull (lifted from
the 2D farthest-point Voronoi diagram).
Problems (iii) and (iv) are max-min or min-max problems, and lack the
symmetry of min-min and max-max problems that enable Eppstein's technique.
For all these problems, the fully dynamic case has remained open.

\PAPER{\paragraph{New results.}}
\LIPICS{\subparagraph*{New results.}}\

\LIPICS{\medskip}
\begin{enumerate}
\item We present the first fully dynamic data structures with sublinear update time for Problems (i)--(iv).  The amortized update time bounds are $O^*(n^{11/12})$ for (i) and (ii), and
$O^*(n^{5/6})$ for (iii) and (iv).  

The approach is general enough to be applicable to many more problems; for example, we can maintain the number of
maximal or ``skyline'' points (points that are not dominated by other points) in a 3D point set in $O^*(n^{2/3})$ amortized time.
%, and can maintain the 
%largest-area Delaunay triangle of a 2D point set in sublinear time.

\LIPICS{\medskip}
\item For basic 3D convex hull queries (e.g., extreme point queries) and 2D nearest neighbor search, as mentioned, Kaplan et al.~\cite{KapSODA17} have lowered the amortized update time of the author's fully dynamic data structure~\cite{SODA06}, from $O(\log^6n)$ to $O(\log^5n)$.
We describe a further logarithmic-factor improvement, from $O(\log^5n)$ to $O(\log^4n)$.  

Although this improvement is admittedly small,
the importance of the result stems from its many applications~\cite{SODA06}; for example,
%we can now maintain the \emph{smallest enclosing circle\/} of a 2D point
%set in $O(\log^4n (\log\log n)^{O(1)})$ time, and 
we can now compute
the \emph{convex (or onion) layers\/} 
of a 3D point set in $O(n\log^4n)$ time, and the \emph{$k$-level\/} in 
an arrangement of planes in 3D in $O(n\log n + f\log^4n)$ time where
$f$ is the output size.

\LIPICS{\medskip}
\item For bichromatic closest pair and diameter in 2D, combining
Eppstein's technique~\cite{Epp} with the above new result on dynamic nearest neighbor search already gives a slightly improved amortized update time of $O(\log^6n)$.  We describe a further, more substantial improvement that eliminates the two extra logarithmic factors caused by Eppstein's technique~\cite{Epp}.
The new update time bound is $O(\log^4n)$.

Dynamic bichromatic closest pair has applications to other problems.
For example, we can now maintain the Euclidean minimum spanning 
tree of a 2D point set with $O(\log^6n)$ amortized update time 
by using another reduction of Eppstein~\cite{Epp} combined with known results
for dynamic minimum spanning trees for graphs~\cite{HolESA15}.
\end{enumerate}

\PAPER{\paragraph{Techniques.}}
\LIPICS{\subparagraph*{Techniques.}}
The common thread in all of our new methods is the use of
\emph{shallow cuttings}: Let $H$ be a set of $n$ hyperplanes in $\R^d$.
The \emph{level\/} of a point $q$ refers to the number of hyperplanes of $H$
strictly below $q$. 
A \emph{$(k,K)$-shallow cutting\/} is a collection of cells covering 
all points of level at most $k$, such that each cell intersects at most $K$ hyperplanes.  The \emph{conflict list\/} $H_\D$ of a cell $\D$ refers to the subset of all hyperplanes of $H$ intersecting $\D$.

Matou\v sek~\cite{MatCGTA92} proved the existence of shallow cuttings with
small number of cells.  Specifically, in 3D, the main lemma can be stated as follows:\footnote{Matou\v sek's original formulation in $\R^d$ states the existence
of a $(k,n/r)$-shallow cutting with $O(r^{\lfloor d/2\rfloor}(1+kr/n)^{\lceil d/2\rceil})$ cells.}

\begin{lemma}\label{lem:shacut}
\emph{(Shallow Cutting Lemma)}\ \ 
Given a set $H$ of $n$ planes in $\R^3$ and a parameter $k\in [1,n]$, 
there exists a $(k,O(k))$-shallow cutting with
$O(n/k)$ cells, where each cell is a ``downward'' tetrahedron
containing $(0,0,-\infty)$. The cutting, together with
the conflict lists of all its cells, can be
constructed in $O(n\log n)$ time.
\end{lemma}

The construction time was first shown by Ramos~\cite{Ram} with
a randomized algorithm.  Later, Chan and Tsakalidis \cite{ChaTsa}
obtained the first $O(n\log n)$-time deterministic algorithm.

To see how static shallow cuttings may be useful for dynamic geometric data structures, observe that most of the
problems considered here are related to the lower envelope of a dynamic
set of planes in $\R^3$ (via duality or the standard lifting transformation).
Usually, the bottleneck lies in
deletions rather than insertions.  
Basically, a shallow cutting provides a compact implicit representation of the
$(\le k)$-level, which is guaranteed to cover the lower envelope even when
up to $k$ deletions have occurred.

A further idea behind all our solutions is to classify planes into two 
types, those that intersect few cells of the shallow cutting, and those that intersect many cells.  The latter type of planes may be bad in slowing down updates, but the key observation is that there can't be too many bad elements.

The new sublinear solutions to Problems (i)--(iv), described 
in Sections~\ref{sec:sublin}--\ref{sec:sublin2}, are obtained by 
incorporating the shallow cutting idea with the previous techniques from
\cite{SODA02}, based on periodic rebuilding.
The entire solution is conceptually not complicated at all, and the description for Problem (i) fits in under two pages, assuming the availability of known range searching structures.
As are typical in other works on data structures with sublinear update time
with ``funny'' exponents,
parameters are judiciously chosen to balance several competing costs. 

The shallow cutting idea has actually been exploited before in dynamic data structures
for basic 3D convex hull queries: Agarwal and Matou\v sek~\cite{AgaMat95} used shallow cuttings recursively (which caused some loss of efficiency), while the author~\cite{SODA06} used a hierarchy of
shallow cuttings, for logarithmically many values of $k$.
The above application of shallow cuttings to Problems (i)--(iv) is even more elementary---we only need a single cutting.  (This makes it all the more embarassing that the idea was missed till now.)

For basic 3D convex hull queries and 2D nearest neighbor search, our improvement is
less innovative.  Described in Section~\ref{sec:dch3d} (which can be read
independently of the previous sections), it is based on the author's original data structure~\cite{SODA06}, with Kaplan et al.'s logarithmic-factor improvement~\cite{KapSODA17}, plus one extra idea to remove
a second logarithmic factor: the main observation is that
Chan and Tsakalidis's algorithm for shallow cuttings~\cite{ChaTsa} already
 constructs an entire hierarchy of  $O(\log n)$ cuttings in $O(n\log n)$ time, not just a single cutting.
However, the hierarchy needed for the data structure in~\cite{SODA06} requires some planes be pruned as we go from one cutting to the next,
so Chan and Tsakalidis's algorithm cannot be applied immediately.  
Still, we show that some nontrivial but technical changes 
(as explained in the appendix) 
can fix the problem.

For 2D bichromatic closest pair and diameter, our $\log^2n$-factor
improvement, described in Section~\ref{sec:bcp}, is a bit
more interesting.  We still do not know how to improve Eppstein's general
reduction~\cite{Epp} from dynamic closest pair to dynamic nearest neighbor search, but intuitively the blind combination of Eppstein's technique
with the author's dynamic data structure for 2D nearest neighbor search seems wasteful, since both share some commonalities (both are sophisticated
variants of the \emph{logarithmic method}~\cite{BenSax}, and both handle deletions via re-insertions of elements into smaller subsets).  To avoid the redundancy, we show how to directly
modify our dynamic data structure for 2D nearest neighbor search to solve the
dynamic 2D bichromatic closest pair problem.  The resulting modification 
completely bypasses Eppstein's ``conga line''
structure~\cite{Epp,Epp00}, and turns out to cause no increase to the $O(\log^4n)$ bound.

\section{Dynamic 3D Convex Hull Size}\label{sec:sublin}

\newcommand{\Hlight}{H_0}%{H_{\mbox{\scriptsize\rm low}}}
\newcommand{\Vlight}{V_0}%{V_{\mbox{\scriptsize\rm low}}}
\newcommand{\Elight}{E_0}%{E_{\mbox{\scriptsize\rm low}}}
\newcommand{\Qlight}{Q_0}%{Q_{\mbox{\scriptsize\rm low}}}
\newcommand{\Hrem}{H_{\mbox{\scriptsize\rm bad}}}
\newcommand{\Trem}{T_{\mbox{\scriptsize\rm bad}}}
\newcommand{\Qrem}{Q_{\mbox{\scriptsize\rm bad}}}
\newcommand{\Jrem}{J_{\mbox{\scriptsize\rm bad}}}

We begin with our new sublinear-time fully dynamic data structure
for maintaining the
number of vertices/facets of the convex hull of a dynamic 3D point set.
The solution is based on the use of shallow cuttings (Lemma~\ref{lem:shacut}) and the author's
previous semi-online data structure~\cite{SODA02}.

\begin{theorem}\label{thm:ch3dsize}
We can maintain the number of vertices, edges, and facets for
the convex hull of a dynamic set of $n$ points in $\R^3$, in general
position, with $O^*(n)$ preprocessing time
and $O^*(n^{11/12})$ amortized insertion and deletion time.
\end{theorem}

\begin{proof}
It suffices to maintain the number of convex hull facets, which determines
the number of vertices and edges (assuming general position).
It suffices to compute the number of upper hull facets, since by symmetry we can compute the number of lower hull facets.
We describe our solution in dual space, where the problem is to
compute the number of vertices in $\LE(H)$ for
a dynamic set $H$ of $n$ planes in $\R^3$.

Let $k$ and $s$ be parameters to be set later.
We divide the update sequence into phases of $k$ updates each.
We maintain a decomposition of the set $H$ into a deletion-only set
$\Hlight$ and a small set $\Hrem$ of ``bad'' planes.

\medskip\noindent{\bf Preprocessing for each phase.}
At the beginning of each phase,
we construct a $(k,O(k))$-shallow cutting $\Gamma$ of $H$
with $O(n/k)$ cells,
together with all their conflict lists, by Lemma~\ref{lem:shacut}.
We set 
\[ \Hlight=\{h\in H: \mbox{$h$ intersects at most $n/s$ cells}\}
\mbox{\ \ and\ \ } \Hrem=H-\Hlight.
\]
Since the total conflict list size is $O(n/k\cdot k)=O(n)$, we have 
$|\Hrem|=O(s)$.

Let $\Vlight$ and $\Elight$ be
the set of vertices and edges of the portion of $\LE(\Hlight)$ 
covered by $\Gamma$, respectively.  There are $O(k)$ such vertices and edges per cell of $\Gamma$, and hence, $|\Vlight|,|\Elight|=O(n/k\cdot k)=O(n)$.
We preprocess $\Vlight$ and $\Elight$ in $O^*(n)$ time by known range searching and intersection searching
techniques, so that 
\LIPICS{\medskip}
\begin{itemize}
\item
we can count the number of points in $\Vlight$
inside a query tetrahedron in $O^*(n^{2/3})$ time (this is 3D simplex range searching)~\cite{Mat92,SoCG10,AgaSURV};
\item
we can count the number of line segments in $\Elight$
intersecting a query triangle in $O^*(n^{3/4})$ time (as noted in~\cite{SODA02}, we can first
solve the case of lines and query halfplanes in $\R^3$ using semialgebraic range searching \cite{AgaMat94}
in Pl\"ucker space,
and then extend the solution for line segments and query triangles
by a multi-level data structure \cite{AgaSURV}).
\end{itemize}
\LIPICS{\medskip}
These data structures can support insertions and deletions of points in $\Vlight$ and line segments in $\Elight$ in $O^*(1)$ time each.
In addition, we preprocess $\Hlight$ in a known dynamic lower envelope data structure
in $O^*(n)$ time, to support ray shooting queries in $\LE(\Hlight)$ in $O^*(1)$ time and deletions in $O^*(1)$ time (e.g., see \cite{AgaMat95} or Section~\ref{sec:dch3d}).
The total preprocessing time per phase is $O^*(n)$.  Amortized over $k$ updates, the cost is $O^*(n/k)$.

\medskip\noindent{\bf Inserting a plane $h$.}
We just insert $h$ to the list $\Hrem$.  Note that $|\Hrem|=O(s+k)$ at all times, since there are at most $k$ insertions per phase.

\medskip\noindent{\bf Deleting a plane $h$ from $\Hrem$.}
We just remove $h$ from the list $\Hrem$.

\medskip\noindent{\bf Deleting a plane $h$ from $\Hlight$.}
We consider
each cell $\D\in\Gamma$ intersected by $h$, and compute
$\LE((\Hlight)_\D)$ from scratch 
in $O(k\log k)$ time (since $|(\Hlight)_\D|=O(k)$).
As the number of cells intersected by $h$ is at most $n/s$,
this computation requires $O^*(kn/s)$ total time.  
The sets $\Vlight$ and $\Elight$ undergo at most $O(kn/s)$ changes,
and their associated data structures can be updated in $O^*(kn/s)$ time.
%The ray shooting structure for $\LE(\Hlight)$ can be updated 
%in $O^*(1)$ time.

\medskip\noindent{\bf Computing the answer.}
To compute the number of vertices of $\LE(H)=\LE(\Hlight\cup\Hrem)$,
we first construct $\LE(\Hrem)$ in $O((s+k)\log(s+k))$ time, and 
triangulate all its $O(s+k)$ faces.  For each triangle $\tau$
in this triangulation:  
\LIPICS{\medskip}
\begin{itemize}
\item 
we count the number of vertices of $\Vlight$ that lie directly below $\tau$, in $O^*(n^{2/3})$ time; and
\item we
count the number of edges of $\Elight$ that intersect $\tau$, in $O^*(n^{3/4})$ time.
\end{itemize}
\LIPICS{\medskip}
We sum up all these counts.
In addition, for each edge of $\LE(\Hrem)$, we test whether
it intersects $\LE(\Hlight)$ by ray shooting in $O^*(1)$ time, and increment the count if true.
For each vertex of $\LE(\Hrem)$, we test whether it is underneath
$\LE(\Hlight)$ by vertical ray shooting in $O^*(1)$ time, and increment the count if true.
Note that $\LE(H)$ is covered by $\Gamma$ at all times, since there are at most $k$ deletions per phase.
The overall count thus gives the answer.  The total time to compute the answer is $O^*((s+k)n^{3/4})$.

\medskip\noindent{\bf Analysis.}
The overall amortized update time is
\[ O^*(n/k + kn/s + (s+k)n^{3/4}).
\]
The theorem follows by setting $s=k^2$ and $k=n^{1/12}$.
\end{proof}

The preprocessing time can be made $O(n\log n)$ and space made $O(n)$ by increasing the update time by an $n^\eps$ factor, via known trade-offs
for range/intersection searching (with larger-degree partition trees).
The method can be deamortized, using existing techniques~\cite{OveBOOK}.

The same method can be adapted to maintain
the sum or maximum of $f(v)$ over all vertices $v$ of $\LE(H)$, for a general class of functions $f$.
Instead of range counting, we store the set $\Vlight$ of points for range sum or range maximum queries
(which have similar complexity as range counting).  For the set $\Elight$ of line segments, the base level of its
multi-level data structure requires data structures ${\cal S}_L$ for
each canonical subset $L$ of lines in $\R^3$, so that we
can return the sum or maximum of $f(\ell\cap h)$ over all $\ell\in L$
for a query plane $h$
in $O^*(|L|^\alpha)$ time, supporting insertions and deletions in $L$ in $O^*(1)$ time.  If $\alpha\le 3/4$, the final time bound of our algorithm remains $O^*(n^{11/12})$.

\begin{theorem}
We can maintain the volume of the convex hull for
a dynamic set of $n$ points in $\R^3$,
with $O^*(n)$ preprocessing time
and $O^*(n^{11/12})$ amortized insertion and deletion time.
\end{theorem}
\begin{proof}
Let $o$ be a fixed point sufficiently far below all the input points.
It suffices to maintain the sum of the volume of the tetrahedra
$op_1p_2p_3$ over all
upper hull facets $p_1p_2p_3$, since by symmetry we can maintain a similar sum for lower hull facets and subtract.
We map each point $p$ to its dual plane $h_p$.
Then the problem fits in the above framework, with
$f(v)$ equal to the volume of the tetrahedron $op_1p_2p_3$ for 
a vertex $v$ defined by the planes $h_{p_1},h_{p_2},h_{p_3}$.
For a fixed line $\ell$ defined by the planes $h_{p_1}$ and $h_{p_2}$, observe that $f(\ell\cap h_p)$ is a linear function over the 3 coordinates of $p$, since the volume of $op_1p_2p$ can be expressed as a determinant.  (This assumes that $op_1p_2$ is oriented clockwise, which we can ensure at the base
level of the multi-level data structure.) 
Thus, we can implement the base structures ${\cal S}_L$ with $\alpha=0$, by simply summing the 4 coefficients of the associated linear functions over all $\ell\in L$.  
\end{proof}

\begin{theorem}
We can maintain the largest empty circle of a dynamic set of $n$ points in $\R^2$, under the restriction that the center lies inside a given triangle $\D_0$, with $O^*(n)$ preprocessing time
and $O^*(n^{11/12})$ amortized insertion and deletion time.
\end{theorem}
\begin{proof}
By the standard lifting transformation,
map each input point $p=(a,b)\in\R^2$ to the plane $h_p$ with equation $z=-2ax-2by+a^2+b^2$ in $\R^3$.  Add 3 near-vertical planes along the edges of $\D_0$.
The largest empty circle problem reduces to finding
a vertex $v=(x,y,z)$ of the lower envelope of these planes, 
maximizing $f(v)=x^2+y^2+z$.
For a fixed line $\ell$, observe that $f(\ell\cap h_{(a,b)})$ is a fixed-degree rational function (ratio of two polynomials) in the 2 variables $a$ and $b$.  
We can implement the base structures ${\cal S}_L$ with $\alpha=2/3$, by
known techniques for semialgebraic range searching in 3D~\cite{AgaMat95} (applied to the graphs of
these bivariate functions).
\end{proof}

We can obtain sublinear update time bounds for
other similar problems, e.g., maintaining the minimum/maximum-area Delaunay triangle of a dynamic 2D point set.  Another application is computing the number of maximal points, also called ``skyline points'' (which are points not dominated by other points), in a dynamic 3D point set:

\begin{theorem}
We can maintain the number of maximal points in a dynamic set $P$ of $n$ points in $\R^3$, with $O^*(n)$ preprocessing time
and $O^*(n^{2/3})$ amortized insertion and deletion time.
\end{theorem}
\begin{proof}
The maximal points are vertices of the upper envelope of orthants
$(-\infty,a]\times (-\infty,b]\times (-\infty,c]$ over all input points $(a,b,c)\in P$ (the upper envelope is an orthogonal polyhedron).  As is well known, an analogue of the shallow cutting lemma holds
for such orthants in 3D (in fact, there is a transformation that maps such orthants to halfspaces in 3D); for example, see~\cite{CLP}.  The same method
can thus be adapted.  In fact, it can be simplified.
The data structure for $\Vlight$ is for orthogonal range searching~\cite{BerBOOK}, which has $O^*(1)$ query and update time.
The data structure $\Elight$ is not needed.
The overall update time becomes 
\[ O^*(n/k + kn/s + (s+k)).
\]
The theorem follows by setting $s=k^2$ and $k=n^{1/3}$.
\end{proof}

We can similarly maintain the volume of a union of $n$ boxes in $\R^3$
in the case when all the boxes have a common corner point at the origin (this is called
the \emph{hypervolume indicator} problem) with $O^*(n^{2/3})$ update time
(previously, an $O^*(\sqrt{n})$ bound was known only in the semi-online setting~\cite{SODA02}).

\section{Dynamic 2D Hausdorff Distance}\label{sec:sublin2}

The method in Section~\ref{sec:sublin} can also be adapted to
solve the dynamic 2D Hausdorff distance problem:

\begin{theorem}
We can maintain the Hausdorff distance between
two dynamic sets $P$ and $Q$ of at most $n$ points in $\R^2$,
with $O^*(n)$ preprocessing time 
and $O^*(n^{8/9})$ amortized insertion and deletion time.
\end{theorem}
\begin{proof}
By the standard lifting transformation,
map each point $p=(a,b)\in P$ to the plane $h_p$ with equation $z=-2ax-2by+a^2+b^2$
in $\R^3$.  Let $H$ be the resulting set of planes.
For each point $q\in Q$, let $\lambda_H(q)$ denote the point on $\LE(H)$ at the vertical line at $q$.
The problem is to find the maximum of $f(\lambda_H(q))$ over all $q\in Q$, where $f(x,y,z)=x^2+y^2+z$, for a dynamic set $H$ of at most $n$ planes and a dynamic set $Q$ of at most $n$ points.

Let $k$ and $s$ be parameters to be set later.
We divide the update sequence into phases of $k$ updates each.
We maintain a decomposition of the set $H$ into a deletion-only set
$\Hlight$ and a small set $\Hrem$ of ``bad'' planes, and 
a decomposition of the set $Q$ into a deletion-only set
$\Qlight$ and a small set $\Qrem$ of ``bad'' points.

\medskip\noindent{\bf Preprocessing for each phase.}
At the beginning of each phase,
we construct a $(k,O(k))$-shallow cutting $\Gamma$ of $H$
with $O(n/k)$ cells,
together with all their conflict lists, by Lemma~\ref{lem:shacut}.  We further subdivide the cells
to ensure that each cell contains at most $k$ points of $Q$ in its $xy$-projection; this can be done by
$O(n/k)$ additional vertical plane cuts, so the number of cells remains
$O(n/k)$.
We set 
\[ \Hlight=\{h\in H: \mbox{$h$ intersects at most $n/s$ cells}\}
\mbox{\ \ and\ \ } \Hrem=H-\Hlight.
\]
Since the total conflict list size is $O(n/k\cdot k)=O(n)$, we have 
$|\Hrem|=O(s)$.  

\newcommand{\lambdalight}{\lambda_{\Hlight}}
We set $Q_0=Q$.  We compute $\lambdalight(q)$ for all $q\in Q$
in $O(n\log n)$ time.   Let $\Lambda_0$ be the 
subset of points in $\{\lambdalight(q) : q\in Q\}$ covered by $\Gamma$.
We preprocess 
the point set $\Lambda_0$ in known 3D simplex range searching data structures  \cite{Mat92,SoCG10,AgaSURV} in $O^*(n)$ time, to support
the following queries in $O^*(n^{2/3})$ time:
\LIPICS{\medskip}
\begin{itemize}
\item compute the maximum of $f(v)$ over 
all points $v\in \Lambda_0$ inside a query tetrahedron;
\item compute the maximum of $f(\lambda_{\{h_p\}}(x,y))$ 
over all points $v=(x,y,z)\in \Lambda_0$ inside a query tetrahedron
for
a query plane $h_p$; note that
maximizing $f(\lambda_{\{h_p\}}(x,y))$ is equivalent to
 maximizing the distance from $(x,y)$ to
$p$ 
(so we can use a 2-level
data structure, combining simplex range searching with 2D farthest neighbor searching).
\end{itemize}
\LIPICS{\medskip}
The data structures can support insertions and deletions of points in $\Lambda_0$ in $O^*(1)$ time each.  In addition, we preprocess $H_0$
in a known dynamic lower envelope data structure in $O^*(n)$ time, to support ray shooting queries in $\LE(H_0)$ in $O^*(1)$ time and deletions in $O^*(1)$ time  (e.g., see \cite{AgaMat95} or Section~\ref{sec:dch3d}).

\medskip\noindent{\bf Inserting a plane $h$ to $H$ or a point $q$ to $Q$.}
We just insert $h$ to the list $\Hrem$ or $q$ to the list $\Qrem$.  Note that $|\Hrem|=O(s+k)$ and $|\Qrem|=O(k)$ at all times.

\medskip\noindent{\bf Deleting a plane $h$ from $\Hrem$ or a point $q$ from $\Qrem$.}
We just remove $h$ from the list $\Hrem$ or $q$ from the list $\Qrem$.

\medskip\noindent{\bf Deleting a point $q$ from $\Qlight$.}
We just remove $\lambdalight(q)$ from the set $\Lambda_0$ in $O^*(1)$ time.

\medskip\noindent{\bf Deleting a plane $h$ from $\Hlight$.}
We consider
each cell $\D\in\Gamma$ intersected by $h$, and compute
$\lambda_{(H_0)_\D}(q)$ for all $q\in Q$ in the $xy$-projection of $\D$ from scratch 
in $O(k\log k)$ time (since $\D$ is intersected by $O(k)$ planes in $H$
and contains $O(k)$ points of $Q$ in its $xy$-projection).
As the number of cells intersected by $h$ is at most $n/s$,
this computation takes $O^*(kn/s)$ total time.
The set $\Lambda_0$ undergoes at most $O(kn/s)$ changes,
and its associated data structures can be updated in $O^*(kn/s)$ time.

\medskip\noindent{\bf Computing the answer.}
To compute the maximum of $f(\lambda_H(q))$ over all $q\in Q$,
we first construct $\LE(\Hrem)$ in $O((s+k)\log(s+k))$ time, and triangulate all its $O(s+k)$ faces.  
For each triangle $\tau$ in this triangulation:
\LIPICS{\medskip}
\begin{itemize}
\item 
We compute the maximum of $f(v)$ over all $v=(x,y,z)\in\Lambda_0$ that lie directly
below $\tau$, in $O^*(n^{2/3})$ time. Note that for all such $v$,
the $\lambda_H(x,y)=\lambda_{\Hlight}(x,y)=v$.
\item We let $h$ be the plane through $\tau$ and compute the maximum of $f(\lambda_{\{h\}}(x,y))$ over all
$v=(x,y,z)\in\Lambda_0$ that lie directly above $\tau$, in $O^*(n^{2/3})$ time.
Note that for all such $v$, $\lambda_H(x,y)=\lambda_{\{h\}}(x,y)$.
\end{itemize}
\LIPICS{\medskip}
In addition, for each $q\in \Qrem$, we compute $\lambda_H(q)$ by vertical
ray shooting in $\LE(\Hlight)$ and $\LE(\Hrem)$ in $O^*(1)$ time; we take the maximum of $f(\lambda_H(q))$
for these points.
Note that $\LE(H)$ is covered by $\Gamma$ at all times, since there are
at most $k$ deletions per phase.
The overall maximum thus gives the answer.
The total time to compute the answer is $O^*((s+k)n^{2/3})$.

\medskip\noindent{\bf Analysis.}
The overall amortized update time is
\[ O^*(n/k + kn/s + (s+k)n^{2/3}).
\]
The theorem follows by setting $s=k^2$ and $k=n^{1/9}$.
\end{proof}

We can similarly solve the dynamic 2D discrete 1-center problem,
by switching lower with upper envelopes and maximum with minimum:

\begin{theorem}
We can maintain the discrete 1-center of
a dynamic set of $n$ points in $\R^2$, with $O^*(n)$ preprocessing time and $O^*(n^{8/9})$ amortized insertion and deletion time.
\end{theorem}

It is possible to slightly improve the $O^*(n^{8/9})$ bound to $O^*(n^{5/6})$ in the preceding two theorems: the key observation is that the point set $\Lambda_0$
is in convex position, and in the convex-position case,
the $O^*(n^{2/3})$ query time for 3D
simplex range searching can be improved to $O^*(\sqrt{n})$, as shown by Sharir and Zaban~\cite{ShaZab}.
(The same observation also improves the author's previous
$O^*(n^{5/6})$ result to $O^*(n^{3/4})$ in the semi-online setting.)

It remains open whether the dynamic Hausdorff distance and discrete 1-center problem in dimensions $d\ge 3$ can similarly be solved in sublinear time.  The author's previous paper~\cite{SODA02} gave an $O^*(n^{1-1/(d+1)(\lceil d/2\rceil + 1)})$-time algorithm but only in the semi-online setting.  In higher dimensions, the size of shallow cuttings becomes too large for the approach to be effective.

\newcommand{\Gin}{\Gamma_{\mbox{\scriptsize\rm in}}}
\newcommand{\Gout}{\Gamma_{\mbox{\scriptsize\rm out}}}

\newcommand{\Vin}{V_{\mbox{\scriptsize\rm in}}}
\newcommand{\Vout}{V_{\mbox{\scriptsize\rm out}}}

\section{Dynamic 3D Convex Hull Queries}\label{sec:dch3d}

In this section, we present a slightly improved data structure for extreme point queries for a dynamic 3D convex hull, by combining the author's previous data structure~\cite{SODA06} (as refined by Kaplan et al.~\cite{KapSODA17}) with a modification of Chan and Tsakalidis's algorithm for constructing a hierarchy of shallow cuttings~\cite{ChaTsa}.  

To describe the latter, we need a definition:
Given a set $H$ of $n$ planes in $\R^3$ and a collection $\Gin$ of cells,
a \emph{$\Gin$-restricted $(k,K)$-shallow cutting} is a collection $\Gout$
of cells covering $\{p\in \R^3: \mbox{$p$ is covered by $\Gin$ and has level at most $k$}\}$, such that each cell in $\Gout$ intersects
at most $K$ planes.  We note that Chan and Tsakalidis's algorithm, with some  technical modifications,
can prove the following lemma.  (The proof requires knowledge of
Chan and Tsakalidis's paper, and is deferred to 
\PAPER{the appendix}%
\LIPICS{the full paper}.)

\begin{lemma}\label{lem:refine:cut}
There exist constants $b$, $c$, and $c'$ such that the following is true:
For a set $H$ of at most $n$ planes in $\R^3$ and a parameter $k\in [1,n]$, given a $(-\infty,cbk)$-shallow cutting\footnote{
In a $(-\infty,k)$-shallow cutting, the cells are not required to cover any particular region.
} 
$\Gin$ with at most $c'n/(bk)$ downward cells, together with their conflict lists,
we can construct a $\Gin$-restricted $(k,ck)$-shallow cutting $\Gout$
with at most $c'n/k$ downward cells, together with their conflict lists,
in $O(n + (n/k)\log (n/k))$ deterministic time.
\end{lemma}

We now redescribe the author's previous data structure~\cite{SODA06} for 3D extreme point queries, with slight changes to incorporate Lemma~\ref{lem:refine:cut}. % (and Kaplan et al.'s improvement).  
The redescription uses a recursive form of the logarithmic method~\cite{BenSax}, which should be a little easier to understand than the original description.

\begin{theorem}
We can maintain a set of $n$ points in $\R^3$, 
with $O(n\log n)$ preprocessing time, $O(\log^2n)$ amortized insertion time, and $O(\log^4n)$ amortized deletion time, so that
we can answer find the extreme point of the convex hull along
any query direction in
$O(\log^2n)$ time.
\end{theorem}
\begin{proof}
We describe our solution in dual space, where we want to answer
vertical ray shooting queries for $\LE(H)$, i.e., find the
lowest plane of $H$ at a query vertical line, for a dynamic set $H$ of 
$n$ planes in $\R^3$.

\medskip\noindent{\bf Preprocessing}.
Our preprocessing algorithm is given by the pseudocode below (ignoring
trivial base cases), with the constants $b,c,c'$ from Lemma~\ref{lem:refine:cut}:\footnote{
Line~4 is where Kaplan et al.'s improvement lies~\cite{KapSODA17}.  The original data structure from~\cite{SODA06} basically had $H_i=H_{i-1} -\{h\in H:\mbox{$h$ intersects more than $2cc'\ell$
cells of $\Gamma_{i}$}\}$.
}

\begin{quote}
\begin{tabbing}
9.M\=\ \ \ \=\ \ \ \=\kill
preprocess($H$):\\[2pt]
%0.\> if $n\le$ a sufficiently large constant then return\\
1.\> $H_0=H$,\ \ $\Gamma_0=\{\R^3\}$,\ \ $\ell=\log_bn$\\
2.\> for $i=1,\ldots,\ell$ do \{\\
3.\>\> $\Gamma_{i}=$ a $\Gamma_{i-1}$-restricted
$(n/b^i, cn/b^i)$-shallow cutting of $H_{i-1}$
with at most $c'b^i$ cells\\
4.
\>\> $H_{i}=H_{i-1}-\{h\in H:\mbox{$h$ intersects more than $2cc'\ell$
cells of $\Gamma_1\cup\cdots\cup\Gamma_{i}$}\}$\\
5.\>\> for each $\D\in\Gamma_{i}$,
compute the conflict list $(H_{i})_\D$ and 
initialize $k_\D=0$\\
  \> \}\\
6.\> preprocess $H_\ell$ for static vertical ray shooting\\
7.\> $\Hrem=H-H_\ell$\\
8.\> preprocess($\Hrem$)
\end{tabbing}
\end{quote}

Note that $\Gamma_{i-1}$ is a $(-\infty,cn/b^{i-1})$-shallow cutting
of $H_{i-2}$, and consequently a $(-\infty,cn/b^{i-1})$-shallow
cutting of $H_{i-1}$, since $H_{i-1}\subseteq H_{i-2}$.  Given
$\Gamma_{i-1}$ and its conflict lists, we can thus apply Lemma~\ref{lem:refine:cut}
to compute $\Gamma_{i}$ and its conflict lists, in $O(n + b^i\log b^i)$
time.  The total time for lines 1--5 is 
$O(n\log n + \sum_{i=1}^{\ell} b^i\log b^i) = O(n\log n)$.
Line~6 takes $O(n\log n)$ time (by a planar point location method~\cite{BerBOOK}).

We claim that $|\Hrem|\le n/2$.  To see this, consider each $h\in \Hrem$.
Let $i$ be the index with $h\in H_{i-1}-H_i$.  Then
$h$ intersects more than $2cc'\ell$ cells
of $\Gamma_1\cup\cdots\cup\Gamma_{i}$; send a charge from $h$ to each
of these cells.  Each cell in $\Gamma_{j}$ receives charges only from planes in $H_{j-1}$ that intersect the cell.  Thus, the total number of charges
is at least $2cc'\ell|\Hrem|$ and is at most $\sum_{j=1}^{\ell} cn/b^j\cdot c'b^j = cc'\ell n$.  The claim follows.
The preprocessing time thus satisfies the recurrence $P(n)\le P(n/2)+O(n\log n)$, which gives $P(n)=O(n\log n)$.
%Space usage satisfies the recurrence $S(n)\le S(n/2)+O(n)$, which gives $S(n)=O(n)$.

\medskip\noindent{\bf Inserting a plane $h$}.
We simply insert $h$ to $\Hrem$ recursively.
When $|\Hrem|$ reaches $3n/4$, we rebuild the data structure for $H$.
It takes $\Omega(n)$ updates for a rebuild to occur.
The amortized insertion time thus satisfies the recurrence
$I(n)\le I(3n/4) + O(P(n)/n) = I(3n/4) + O(\log n)$,
which gives $I(n)=O(\log^2n)$.

\medskip\noindent{\bf Deleting a plane $h$}.  The deletion algorithm
is as follows:
\PAPER{\newpage}
\begin{quote}
\begin{tabbing}
9.M \=\ \ \ \=\ \ \ \=\ \ \ \ \=\kill
delete($H,h$):\\[2pt]
1.\> for $i=1,\ldots,\ell$ do\\
2.\>\> for each $\D\in\Gamma_{i}$ with
$h\in (H_{i})_\D$ do \{\\
3.\>\>\> increment $k_\D$\\
4.\>\>\> if $k_\D \ge n/b^{i+1}$ then\\
5.\>\>\>\> for all $h\in (H_i)_\D$ that are still in $H$ but not yet in $\Hrem$, 
insert $h$ to $\Hrem$\\
\> \}\\
6.\> if $h\in\Hrem$ then delete($\Hrem,h$)
\end{tabbing}
\end{quote}

Let $i$ be the largest index with $h\in H_{i}$.  Then $h$ intersects at most 
$2cc'\ell=O(\log n)$ cells of $\Gamma_1\cup\cdots\cup\Gamma_{i}$.
Thus, in each deletion, lines 3--5 are executed $O(\log n)$ times.

In lines 3--5, it takes $n/b^{i+1}$ increments of $k_\D$
to cause the $|(H_{i})_\D|\le cn/b^{i}$ planes to be inserted to $\Hrem$.
Thus, each increment triggers $O(1)$ amortized number of insertions to $\Hrem$, and so a deletion triggers $O(\log n)$ amortized number of insertions to $\Hrem$.
The amortized deletion time thus satisfies the recurrence
$D(n)\le D(3n/4) + O(\log n) I(3n/4) = D(3n/4) + O(\log^3 n)$,
which gives $D(n)=O(\log^4n)$.

\medskip\noindent{\bf Answering the query for a vertical line $q$}.
We first answer the query for the static set $H_\ell$
in $O(\log n)$ time (by planar point location); if the returned plane has already been deleted, ignore the answer.
We then recursively answer the query for $\Hrem$, and return
the lowest of all the planes found.
The query time satisfies the recurrence $Q(n)\le Q(3n/4)+O(\log n)$,
which gives $Q(n)=O(\log^2n)$.

\medskip\noindent{\bf Correctness of the query algorithm.}
To prove correctness, let $h^*$ be the lowest plane at $q$ and
$v^*=h^*\cap q$.
If $h^*\in\Hrem$, correctness follows by induction.  
So, assume that $h^*\not\in\Hrem$.

If $v^*$ is covered by $\Gamma_\ell$, say, by the cell $\D\in \Gamma_\ell$, then
either $v^*$ is on $\LE(H_\ell)$, in which case the algorithm would have correctly
found $h^*$, or some plane in $(H_\ell)_\D$
has been deleted from $H$, in which case all active planes of $(H_\ell)_\D$, including $h^*$, would have been inserted to $\Hrem$.

Otherwise, let $i$ be an index such that $v^*$ is not covered by $\Gamma_i$ but is covered by $\Gamma_{i-1}$, say, by the cell $\D\in\Gamma_{i-1}$.  Since $\Gamma_{i}$ is a $\Gamma_{i-1}$-restricted
$(n/b^{i}, cn/b^{i})$-shallow cutting of $H_{i-1}$,
it follows that $v^*$ must have level more than $n/b^{i}$ in $H_{i-1}$.
In order for $v^*$ to be the answer,
the more than $n/b^{i}$ planes of $H_{i-1}$ below $v^*$ must
have been deleted from $H$.  But then all active planes of
$(H_{i-1})_\D$, including $h^*$, would have been inserted to $\Hrem$.
\end{proof}

By the standard lifting transformation, we obtain:

\begin{corollary}
We can maintain a set of $n$ points in $\R^2$, 
with $O(n\log n)$ preprocessing time, $O(\log^2n)$ amortized insertion time, and $O(\log^4n)$ amortized deletion time, so that
we can answer find the nearest neighbor to any query point in
$O(\log^2n)$ time.
\end{corollary}

The space usage in the above data structure is $O(n\log n)$, but can
be improved to $O(n)$, by following an idea mentioned in~\cite{SODA06}
(due to Afshani): instead
of storing conflict lists explicitly, generate conflict lists on demand by using a known optimal (static) linear-space data structure for
halfspace range reporting~\cite{AfsCha}.

Following~\cite{SODA06}, we can use the same dynamic data structure to answer other
basic types of 3D convex hull queries, e.g., \emph{gift wrapping queries}
(finding the two tangents of the hull with a query line outside the hull) in
$O(\log^2n)$ time and \emph{line-intersection queries}
(intersecting the hull with a query line) 
in $O(\log^4 n \log^{O(1)}\log n)$ time.  The latter corresponds to
\emph{3D linear programming queries} in dual space.  The dynamic data structure
can be adapted to maintain the \emph{smallest enclosing circle} of a 2D point set.
Following~\cite{SoCG11}, the dynamic data structure can also be adapted to answer 
3D halfspace range reporting queries.

\section{Dynamic 2D Bichromatic Closest Pair}\label{sec:bcp}

We now adapt the data structure in Section~\ref{sec:dch3d} to
solve the dynamic 2D bichromatic closest pair problem:

\begin{theorem}
We can maintain the closest pair between
two dynamic sets $P$ and $Q$ of at most $n$ points in $\R^2$,
with $O(n\log n)$ preprocessing time, $O(\log^2n)$ amortized insertion time, and $O(\log^4n)$ amortized deletion time.
\end{theorem}
\begin{proof}
By the standard lifting transformation,
map each input point $p=(a,b)$ to the plane $h_p$ with equation $z=-2ax-2by+a^2+b^2$
in $\R^3$.  Let $H=\{h_p:p\in P\}$.
For each point $q\in Q$, let $\lambda_H(q)$ denote the point on $\LE(H)$ at the vertical line at $q$.  
Let $J=\{h_q:q\in Q\}$.
For each point $p\in P$, define $\lambda_J(p)$ similarly.
We want to compute the minimum of $f(\lambda_H(q))$ over all $q\in Q$,
where $f(x,y,z)=x^2+y^2+z$, which is equivalent to the minimum of 
$f(\lambda_J(p))$ over all $p\in P$.

\medskip\noindent{\bf Preprocessing}.
We maintain a global heap, whose minimum gives the answer.
We modify the preprocess($H$) algorithm in Section~\ref{sec:dch3d}:
\begin{quote}
\begin{tabbing}
9.M\=\ \ \ \=\ \ \ \=\kill
preprocess($H,J$):\\[2pt]
1.\> run lines 1--7 of the preprocess($H$) algorithm 
on $H$\\
2.\> for each $h_q\in J$, add $f(\lambda_{H_\ell}(q))$ to the heap\\
3.\> run lines 1--7 of the preprocess($H$) algorithm 
but with $H$'s replaced by $J$'s\\
4.\> for each $h_p\in H$, add $f(\lambda_{J_\ell}(p))$ to the heap\\
5.\> preprocess($\Hrem,\Jrem$)
\end{tabbing}
\end{quote}

As in Section~\ref{sec:dch3d}, the preprocessing time satisfies the recurrence
$P(n)\le P(n/2)+O(n\log n)$, which gives $P(n)=O(n\log n)$.

\medskip\noindent{\bf Inserting a plane $h_p$ to $H$}.
We recursively insert $h_p$ to $\Hrem$. 
We also compute $\lambda_{J_\ell}(p)$ in $O(\log n)$ time (by planar point location), and add $f(\lambda_{J_\ell}(p))$ to the heap.
  
When $|\Hrem|$ or $|\Jrem|$ reaches $3n/4$, we rebuild the data structure for $H$ and $J$.
It takes $\Omega(n)$ updates for a rebuild to occur.
The amortized insertion time thus satisfies the recurrence
$I(n)\le I(3n/4) + O(\log n) + O(P(n)/n) = I(3n/4) + O(\log n)$,
which gives $I(n)=O(\log^2n)$.

\medskip\noindent{\bf Inserting a plane $h_q$ to $J$}.
Symmetric to the above.

\medskip\noindent{\bf Deleting a plane $h_p$ from $H$}.
We run lines 1--5 of the delete($H,h$) algorithm in Section~\ref{sec:dch3d} (with $h=h_p$).
In the heap, we remove all entries $f(\lambda_{H_\ell}(q))$ that has
$\lambda_{H_\ell}(q)=\lambda_{\{h_p\}}(q)$.  
If $h_p\in \Hrem$, we further recursively delete $h_p$ from $\Hrem$.
We also remove $f(\lambda_{J_\ell}(p))$ from the heap.

For the analysis, we can charge removals of entries from the heap to their corresponding insertions, by
amortization. 
The amortized deletion time thus satisfies the recurrence
$D(n)\le D(3n/4) + O(\log n) I(3n/4) = D(3n/4) + O(\log^3 n)$,
which gives $D(n)=O(\log^4n)$.

\medskip\noindent{\bf Deleting a plane $h_q$ from $J$}.
Symmetric to the above.

\medskip\noindent{\bf Correctness}.
Let $p^*q^*$ be the closest pair with $p^*\in P$ and $q^*\in Q$.
If both $h_{p^*}\in \Hrem$ and $h_{q^*}\in \Jrem$, correctness follows
by induction.
Otherwise, assume without loss of generality that $h_{p^*}\not\in \Hrem$.  
(The case $J_{q^*}\not\in \Jrem$ is symmetric.)  Let $v^*=\lambda_H(q^*)$.
The rest of the correctness argument is essentially identical to that in Section~\ref{sec:dch3d}:

If $v^*$ is covered by $\Gamma_\ell$, say, by the cell $\D\in \Gamma_\ell$, then
either $v^*$ is on $\LE(H_\ell)$, in which case the algorithm would have included $f(\lambda_H(q^*))$ in the heap, or some plane in $(H_\ell)_\D$
has been deleted from $H$, in which case all active planes of $(H_\ell)_\D$, including $h_{p^*}$, would have been inserted to $\Hrem$.

Otherwise, let $i$ be an index such that $v^*$ is not covered by $\Gamma_i$ but is covered by $\Gamma_{i-1}$, say, by the cell $\D\in\Gamma_{i-1}$.  Since $\Gamma_{i}$ is a $\Gamma_{i-1}$-restricted
$(n/b^{i}, cn/b^{i})$-shallow cutting of $H_{i-1}$,
it follows that $v^*$ must have level more than $n/b^{i}$ in $H_{i-1}$.
In order for $v^*$ to be the answer,
the more than $n/b^{i}$ planes of $H_{i-1}$ below $v^*$ must
have been deleted from $H$.  But then all active planes of
$(H_{i-1})_\D$, including $h_{p^*}$, would have been inserted to $\Hrem$.
\end{proof}

We can similarly solve the diameter problem, by replacing min with max and lower with upper envelopes:

\begin{theorem}
We can maintain the diameter of a dynamic set of $n$ points in $\R^2$, 
with $O(n\log n)$ preprocessing time, $O(\log^2n)$ amortized insertion time, and $O(\log^4n)$ amortized deletion time. 
\end{theorem}

\LIPICS{\subparagraph*{Acknowledgement.}}
\PAPER{\paragraph{Acknowledgement.}}
I thank Sariel Har-Peled for discussions on other problems that indirectly led to the results of this paper.  Thanks also to Mitchell Jones for discussions on range searching for points in convex position.

{\small
\bibliographystyle{plainurl}%{abbrv}
\bibliography{dyn_sha_cut}
}

\PAPER{

\appendix
\section{Proof of Lemma~\ref{lem:refine:cut}}

%\begin{proof}
As in the previous paper~\cite{ChaTsa}, it is more convenient to
work with shallow cuttings in vertex form:
given a set $H$ of $n$ planes in $\R^3$,
a \emph{$(k,K)$-shallow cutting in vertex form} is 
a set $V$ of points whose upper hull $\UH(V)$ covers all points in $\R^3$ of level
at most $k$, such that every point in $V$ has level at most $K$.
The \emph{conflict list} of a point refers to the list of all planes of $H$
below the point.  

Chan and Tsakalidis \cite[Theorem 5]{ChaTsa}
proved the following statement for some constants $B$, $C$, and $C'$:
\begin{quote}
For a set $H$ of at most $n$ planes in $\R^3$ and a parameter
$k\in [1,n]$, given a $(Bk, CBk)$-shallow cutting
$\Vin$ in vertex form with at most $C'n/(Bk)$ vertices, together
with their conflict lists, 
we can construct a $(k,Ck)$-shallow cutting
$\Vout$ in vertex form with at most $C'n/k$ vertices, together
with their conflict lists,
in $O(n + (n/k)\log (n/k))$ deterministic time.
\end{quote}

By a close inspection of their proof, we actually get
the following stronger statement for some absolute constants $a_0'$, $c_0$, and $c_0'$,
for \emph{any} choice of constants $B$, $C'$, and $t$:
\begin{quote}
For a set $H$ of at most $n$ planes in $\R^3$ and a parameter
$k\in [1,n]$, 
given a $(12c_0^2k, CBk)$-shallow cutting
$\Vin$ in vertex form with at most $C'n/(Bk)$ vertices, together
with their conflict lists, 
we can construct a $(k,Ck)$-shallow cutting
$\Vout$ in vertex form with at most 
$C''n/k$ vertices, together
with their conflict lists,
in $O(n + (n/k)\log (n/k))$ deterministic time, where
$C=12c_0^2+1$ and $C''= 2c_0' + \frac{8a_0'c_0'CC'}{3c_0\sqrt{t}}$.
\end{quote}

\newcommand{\HHH}{\widehat{H}}
\newcommand{\nnn}{\widehat{n}}
Set $b=B/2$, $c=3C$, and $c'=C'/36$.
To derive Lemma~\ref{lem:refine:cut} from the above statement,
we first convert $\Gin$ to vertex form, by letting $\Vin$ to be
the vertices of the upper hull of the cells in $\Gin$.  
Then $|\Vin|\le 3c'n/(bk)$.  It is helpful to
assume that each vertex in $\Vin$ has degree 3 in the upper hull;
this can be guaranteed by intersecting the upper hull
with extra planes infinitesimally close to each vertex (the number of new vertices is equal to
twice the number of old edges, i.e., at most 6 times the number of old vertices).  After this modification, $|\Vin|\le 18c'n/(bk)=C'n/(Bk)$.

We can't apply the above result to $H$ directly.  Instead,
we make $12c_0^2k$ copies of the plane 
through each facet of $\UH(\Vin)$, and add them to $H$.
Since there are at most $2|\Vin|$ such facets, the new set $\HHH$ has size $\nnn\le n + (12c_0^2k)C'n/(Bk)\le 2n$, by setting $B= 24c_0^2C'$.  Then
$\UH(\Vin)$ covers all points of level at most $12c_0^2k$ in $\HHH$.
Each point in $\Vin$ has level at most $Cbk + 3(12c_0^2k)\le CBk$ in $\HHH$,
assuming $B\ge 72c_0^2$.  

We can now apply the above to $\HHH$, and
obtain a $(k,Ck)$-shallow cutting $\Vout$ for $\HHH$ in vertex form.
We can set $\Gout$ to be the vertical decomposition of $\UH(\Vout)$
(which has at most $2|\Vout|$ facets).
Each cell of $\Gout$ intersects at most $3Ck=ck$ planes of $H$.    Every point covered by $\Gin$ with level at most $k$
in $H$ has level at most $k$ in $\HHH$ and is thus covered by $\Gout$.
Furthermore, $|\Gout|\le 2C''\nnn/k\le 4C''n/k=
4(2c_0'+ \frac{8a_0'c_0'CC'}{3c_0\sqrt{t}})n/k\le 12c_0'n/k$
by setting $t=(\frac{8a_0'CC'}{3c_0})^2$.
Thus, $|\Gout|\le C'n/k$ by setting
$C'=36\cdot 12c_0'$.
\qed
%\end{proof}

}

\end{document}